\def\AP            {{\ensuremath{A_{\mathcal P}}}}
\def\be            {\begin{equation}}
\def\bearl         {\begin{array}{l}}
\def\bearll        {\begin{array}{ll}}
\def\bimod         {\mbox{\rm -bimod}}
\def\boti          {\,{\boxtimes}\,}
\def\C             {{\ensuremath{\mathcal C}}}

\def\calc          {{\ensuremath{\mathcal C}}}
\def\cald          {{\ensuremath{\mathcal D}}}
\def\calm          {{\ensuremath{\mathcal M}}}
\def\calw          {{\ensuremath{\mathcal W}}}

\def\cir           {\,{\circ}\,}
\def\cobord        {\mathrm{cobord}}
\def\complex       {{\ensuremath{\mathbb C}}}

\def\cov           {\mathrm{cov}}

\def\dim           {\mathrm{dim}}
\def\dimc          {\dim_\complex}
\def\ee            {\end{equation}}

\def\eear          {\end{array}}
\def\End           {{\ensuremath{\mathrm{End}}}}

\def\env           {{\mathrm{env}}}

\def\erf           {\eqref }
\def\eq            {\,{=}\,}

\def\Fun           {{\ensuremath{\mathrm{Fun}}}}

\newcommand\hsp[1] {\mbox{\hspace{#1 em}}}
\def\id            {\mbox{\sl id}}
\def\Id            {\mbox{\sl Id}}

\def\idsm          {\mbox{\footnotesize\sl id}}
\def\Idsm          {\mbox{\footnotesize\sl Id}}

\def\iN            {\,{\in}\,}

\newcommand\labl[1]{\label{#1}\ee}

\def\Mod           {\mbox{\rm -mod}}

\def\mti           {\,{\mtimes}\,}
\def\mtimes        {\odot}  

\newcommand\Nxl[1] {\\[-1.3em]\\[#1mm]}

\def\ohr           {\varrho}
\def\one           {{\bf1}}
\def\oned          {{\bf1_{\mathcal D}}}
\def\op            {^{\mathrm{op}}}
\def\oti           {\,{\otimes}\,}
\def\Oti           {{\otimes}}
\def\P             {{\ensuremath{\mathcal P_{\mathcal D}}}}
\def\Po            {{\mathcal P}}
\def\pop           {Y}

\def\rev           {{\mathrm{rev}}}
\def\Sigma         {\varSigma}
\def\Spq           {\ensuremath{S_{\vec p,\vec q}}} 
\def\sse           {\scriptsize }
\def\T             {{\ensuremath{\mathcal T}}}
\def\tft           {\mathrm{tft}}
\def\threedim      {three-di\-men\-si\-o\-nal}
\def\Times         {\,{\times}\,}
\def\To            {\,{\to}\,}
\def\ucalc         {_{\mathcal C}}     
\def\ucald         {_{\mathcal D}}     
\def\Vect          {\ensuremath{\mathrm{vect}}}

 \newcommand\void[1]{}
\def\VVect         {\ensuremath 2\mbox{-}\ensuremath{\mathrm{vect}}}
\def\xymapsto      {\!\xymatrix@C=0.5cm{\ar@{|->}[r] &}\!}
\def\xyRightarrow  {\!\xymatrix@C=0.5cm{\ar@{=>}[r] &}\!}
\def\xyRightarroww {\!\xymatrix@C=2.8em{\ar@{=>}[r] &}\!}
\def\xyto          {\!\xymatrix@C=0.5cm{\ar[r] &}\!}
\def\Z             {{\ensuremath{\mathcal Z}}}
\def\zet           {{\ensuremath{\mathbb Z}}}

\documentclass[12pt]{article}
\usepackage{latexsym, amsmath, amsthm, amsfonts} 
\usepackage{enumerate, amssymb, xspace, xypic}
\usepackage[all]{xy}
\usepackage[mathscr]{eucal}
\usepackage{graphicx} \usepackage{rotating}
\usepackage{epstopdf,hyperref}
\usepackage{color}

\newtheorem{thm}{Theorem}

\newtheorem{lem}[thm]{Lemma}

\newtheorem{prop}[thm]{Proposition}

\theoremstyle{definition}

\newtheorem{rem}[thm]{Remark}

\newcommand\eqpic[4]{\begin{eqnarray}
                   \begin{picture}(#2,#3){}\end{picture}\nonumber\\
                   \raisebox{-#3pt}{ \begin{picture}(#2,#3) #4 \end{picture} }
                   \label{#1} \\~\nonumber \end{eqnarray} }

\newcommand\Includepicfj[2]   {{\begin{picture}(0,0)(0,0)
                            \scalebox{.#1}{\includegraphics{#2.eps}}\end{picture}}}

\begin{document}

\def\cir{\,{\circ}\,} 
                                     
\begin{flushright}
   {\sf ZMP-HH/13-18}\\
   {\sf Hamburger$\;$Beitr\"age$\;$zur$\;$Mathematik$\;$Nr.$\;$490}\\[2mm]
October 2013
\end{flushright}
\vskip 3.5em
                                   
\begin{center}
\begin{tabular}c \Large\bf 
A note on permutation twist defects \\[3mm] \Large\bf 
in topological bilayer phases
\end{tabular}\vskip 2.6em
                                     
  ~J\"urgen Fuchs\,$^{\,a}$ ~~~and~~~
  ~Christoph Schweigert\,$^{\,b}$

\vskip 9mm

  \it$^a$
  Teoretisk fysik, \ Karlstads Universitet\\
  Universitetsgatan 21, \ S\,--\,651\,88\, Karlstad \\[7pt]
  \it$^b$
  Fachbereich Mathematik, \ Universit\"at Hamburg\\
  Bereich Algebra und Zahlentheorie\\
  Bundesstra\ss e 55, \ D\,--\,20\,146\, Hamburg
                   
\end{center}
                     
\vskip 5.3em

\noindent{\sc Abstract}\\[3pt]
We present a mathematical derivation of some of the most important physical 
quantities arising in topological bilayer systems with permutation
twist defects as introduced by Barkeshli et al.\ \cite{bajQ2}.
A crucial tool is the theory of permutation equivariant modular functors 
developed by Barmeier et al.\ \cite{bfrs2,barSc}.

\newpage

\section{Introduction}\label{intro}

Topological phases of matter have become one of the most important areas of 
interplay between mathematical physics and condensed matter physics. It has 
been known for a long time that three-dimensional topological field theories 
of Reshetikhin-Turaev type describe certain universality classes of robustly
gapped systems in 2+1 dimensions, e.g.\ of quantum Hall fluids. This class of
theories includes in particular abelian Chern-Simons theories, which can be
defined using integer lattices. In more generality, a basic datum describing 
a topological phase is a modular tensor category (this notion will be
recalled in Section \ref{sec:2}).

We wish to study a topological phase described by some modular tensor category
\cald, exclusively working in the context of the category \cald. From
\cald\ one can construct an extended topological field theory, as a
symmetric monoidal 2-functor
  \be
  \tft_\cald:\quad \cobord_{3,2,1} \longrightarrow \VVect 
  \label{ord:cft}\ee
from extended cobordisms to 2-vector spaces, i.e.\ to finitely semisimple
abelian categories; we refer to \cite{mort6} for a 
review of the relevant notions. 
The 2-functor $\tft_\cald$ provides us with the following structure.
  \def\leftmargini{1.27em}~\\[-1.55em]\begin{itemize}\addtolength{\itemsep}{-7pt}

\item
For a one-dimensional oriented manifold $S$ one gets a category: for
$S$ the disjoint union of $n$ circles, one has 
$\tft_\cald(S) \eq \cald^{\boxtimes n}$, with $\boxtimes$ the Deligne product 
of categories enriched over complex vector spaces. The category $\cald$ 
associated to a circle has the physical interpretation of labels for Wilson 
lines and for point-like insertions on Wilson lines. Thus in the condensed 
matter system the isomorphism classes of objects of \cald\ describe types of 
quasi-particle excitations in the topological phase.

\item
For a surface, possibly with marked points -- or rather, closed disks excised
around the marked points -- that are end points of Wilson lines and thus
carry the labels of quasi-particles, we get a vector space of conformal blocks.
In the context of topological quantum computing these spaces, as subspaces 
of vector spaces arising in a suitable microscopical model, have been proposed 
for quantum codes; see e.g.\ \cite{bmca,kiKon} for details in the context of 
Hopf algebras and of lattice models corresponding to TFTs of Turaev-Viro type, 
respectively.
\\
(In the present paper we are only interested in the behavior
of universality classes. We do not touch the very important question
of their possible realization in microscopic models.)

\item
Three-dimensional manifolds with corners give linear maps between the spaces 
of conformal blocks assigned to their boundaries. On the category \cald\ this 
amounts in particular to the structure of a braiding, and on the vector spaces 
of conformal blocks to representations of mapping class groups. 
The latter are of interest in the implementation of quantum gates. For example, 
one needs information about the `size' of the representations of the mapping 
class group to know whether a given system allows for universal quantum gates.

\end{itemize}
This structure raises in particular three types of 
questions about topological phases: 
  \def\leftmargini{1.27em}~\\[-1.55em]\begin{itemize}\addtolength{\itemsep}{-7pt}

\item
The problem of classifying the possible types of quasi-particles.

\item
The problem of computing the dimension of the associated spaces of conformal 
blocks, which has the interpretation of the number of qubits that can be 
stored in the corresponding quantum code.

\item
The problem of understanding the braid group
representations on spaces of conformal blocks.
\end{itemize}
 
In recent years physical boundaries and surface defects in three-dimensional
topological field theories have attracted increasing attention.
In this note we do not consider boundaries; we furthermore restrict 
our attention to surface defects in a single topological phase described
by the modular tensor category \cald.

Surface defects lead to a riches of phenomena. Specifically,
the three types of problems just stated have natural generalizations to 
situations in which defects are present. These are the issues we discuss in 
the present note, for the specific case of twist defects in bilayer systems.
  \def\leftmargini{1.27em}~\\[-1.55em]\begin{itemize}\addtolength{\itemsep}{-7pt}

\item
For any pair $a,\,a'$ of topological surface defects there is a category
$\calw_{a,a'}$ of Wilson lines confined to the defect surface, which
separate a surface region labeled by $a$ from a surface region labeled by $a'$.
The objects of this category label defect Wilson lines, while the morphisms
label point-like insertions on those Wilson lines. In the application to
topological phases, such Wilson lines do not describe intrinsic
quasi-particles, but rather extrinsic objects with long-range interaction
\cite{bajQ2}. More generally, there are
categories of Wilson lines at which an arbitrary  finite number of surface
defects end (see the figure \erf{Figure13+Figure14} below).
A first problem is to obtain a concrete description of these categories.

\item
Surface defects and generalized Wilson lines can intersect transversally 
surfaces to which one would like to associate appropriate generalizations of
conformal blocks. This raises the question of how to define these vector spaces
and how to obtain expressions for
their dimensions, generalizing the Verlinde formula.

\item
On these vector spaces, one has the action of appropriate versions of mapping 
class groups. It is an important task to understand these groups and their 
actions in detail.

\end{itemize}

It has been demonstrated in a model-independent analysis \cite{fusV} that 
surface defects which separate two regions that are both in the phase labeled 
by \cald\ are described by module categories over the modular tensor category 
\cald\ (again, the notion of \cald-module category, or \cald-module, will be 
recalled in Section \ref{sec:2}). For a general modular tensor category
\cald, the classification of \cald-modules, and thus of surface defects,
is out of reach. Notable exceptions are abelian Chern-Si\-mons theories, as 
studied in \cite{kaSau2,fusV,bajQ3}, and Dijkgraaf-Witten theories,
for which \cald\ is the category of finite-dimensional representations of the 
Drinfeld double of a finite group $G$, as well as the category of integrable 
highest weight representations of the affine Lie algebra based on 
$\mathfrak{sl}(2,\complex)$ at positive integral level. For Dijkgraaf-Witten
theories subgroups of $G$ and certain group cochains enter the classification 
\cite{ostr5} (for the corresponding boundary conditions and surface defects 
and their geometric interpretation see \cite{fusV2}), while in the 
$\mathfrak{sl}(2,\complex)$ case an A-D-E classification emerges \cite{kios}.

Any modular tensor category \cald\ has at least one indecomposable module,
namely the regular \cald-module given by the abelian category \cald\ itself, 
with the action being just the ordinary tensor product of \cald. The 
corresponding defect $\T_\cald$ is the \emph{transparent} (or invisible) defect,
which is 
a monoidal unit under fusion of surface defects. (For a discussion of $\T_\cald$
in the context of Dijkgraaf-Witten models see \cite[Sect.\,3.6]{fusV2}.)
Generically it can be hard to find other \cald-mod\-u\-les besides $\T_\cald$.
There is, however, one general situation in which a non-trivial module 
category can be identified, namely when the modular tensor category \cald\ 
is the Deligne product $\cald \eq \calc\boti\calc$ of two copies of another 
modular tensor category \calc\ (with the same structure of modular tensor 
category, in particular the same braiding, on each copy). In the condensed 
matter literature this situation is known as a \emph{bilayer system},
see e.g.\ \cite{baWen3,baQi2,bajQ2} for recent work.   

For any bilayer system described by $\cald \eq \calc\boti\calc$, in addition to
the regular \cald-module there is a second module category, which is defined 
on a single copy of the abelian category \calc; its module category structure 
has been given in \cite[Thm.\,2.2]{bfrs2}. This type of surface defect is 
the main subject of the present letter. We will present it in detail in Section 
\ref{sec:2}, where we also demonstrate that it corresponds to the permutation 
twist defect \P\ \cite{bajQ2} of the bilayer system, which permutes the 
different layers. We then also obtain the relevant categories of generalized 
Wilson lines, thereby providing a model-independent solution of
the problem of describing the generalized 
quasi-par\-ticle excitations in the presence of the permutation twist defect. 
In Section \ref{sec:3} we present the relevant spaces of conformal blocks. In 
applications to quantum computing their dimensions give the number of qubits 
that can be stored in a topological code realized by bilayer fractional quantum 
Hall states in the universality class of the bilayer topological phase.

The results reported in this letter have obvious generalizations to $n$-layer
systems, corresponding to the situation that $\cald \eq \calc^{\boxtimes n}$ is 
the Deligne product of any number $n$ of copies of the same modular tensor 
category \calc. In this case we find a \cald-module category for every element
of the symmetric group $S_n$ that describes a permutation of the $n$ layers. 
The mathematical tools to describe these systems in detail are available
\cite{bohs,bifs,bant8,bfrs2,barSc}: for the sake of clarity of
the exposition we refrain from presenting them in this note.


\section{The module category {\boldmath \P} and categories of defect Wilson lines}
\label{sec:2}

\paragraph{Module categories.}
 
A \emph{fusion category} (over \complex) is a rigid semisimple linear monoidal 
category, enriched over the category of (complex) vector spaces, with only 
finitely many isomorphism classes of simple objects, such that the endomorphisms 
of the monoidal unit form just the ground field, $\End(\one) \eq \complex\,
\id_\one$. All categories in this letter will be finitely semisimple abelian
$\complex$-linear categories. Usually we assume fusion categories to be 
strictly monoidal, so that we can drop the associativity constraints. All 
categories of Wilson lines in topological field theories, including defect and 
boundary Wilson lines, are in our context described by fusion categories.
A \emph{modular tensor category} \calc\ is a braided fusion category in which
the braiding $c_{U,V}\colon U\oti V \,{\xrightarrow{\,\cong\,}}\, V\oti U$ obeys
a non-degeneracy condition: the $|I_\calc|\Times|I_\calc|$-matrix with entries 
  \be
  s_{i.j}^{} := \mathrm{tr}(c^{}_{S_j,S_i}\cir c^{}_{S_j,S_i}) \,, 
  \labl{smat}
where $(S_i)_{i\in I_\calc}$ is 
a set of representatives for the isomorphism classes of simple objects of $\calc$,
is invertible.
Examples of modular tensor categories arise from Chern-Simons theories, in which
case they can be described in terms of finite abelian groups with a quadratic
form (abelian Chern-Simons theories) or of integrable highest weight 
representations of affine Lie algebras (non-abelian Chern-Simons theories).

A \emph{module category} \calm\ over a monoidal category \cald\ (or a
\cald-module, for short) consists of a category \calm\ and  a \complex-linear 
bifunctor $\mtimes\colon \cald \Times \calm \To \calm$ together with functorial 
associativity and unit isomorphisms
  \be
  (X\oti Y)\mtimes M \stackrel\cong\longrightarrow X \mtimes (Y\mti M)
  \qquad\text{and}\qquad
  \one\mtimes M \stackrel\cong\longrightarrow M 
  \ee
for $X,Y\iN\cald$ and $M\iN \calm$, obeying coherence conditions.
For details see e.g.\ \cite[Sect.\,2.3]{ostr}; 
we refer to the functor $\mtimes$ as the mixed tensor product. Thinking about
a fusion category as a categorification of a unital associative ring,
module categories are categorifications of modules over that ring. Taking the 
tensor product $\otimes\colon \cald\Times\cald \To \cald$ and the associated 
associativity and unit constraints of \cald, any fusion category \cald\ is, as 
already mentioned, a module category $\T_\cald$ over itself, analogously as 
any ring is a module over itself. In topological field theories with defects 
the module category $\T_\cald$ describes the transparent defect, separating 
two regions that both support the same topological phase labeled by \cald.

\paragraph{A module category for bilayer systems.}
 
In the case of bilayer systems, i.e.\ $\cald \eq \calc\boti\calc$ with
a modular tensor category \calc, there is another generic module category \P.
The abelian category underlying \P\ is just \calc\ itself. The mixed tensor 
product for \P\ is defined in terms of the tensor product $\otimes$ in \calc\ by
  \be
  (U\boti V) \mtimes M =  U \oti V \oti M
  \ee
for $M\iN\calc$ and $U\boti V\iN\cald$. (In the definition, it suffices to 
consider objects of $\cald$ of the form $U\boti V$ 
with $U,V\iN\calc$ only; such objects are called $\boxtimes$-factorizable.)
        
The existence of such a \cald-module generalizes the fact that for a 
{\em commutative} ring $R$, the tensor product ring $R \,{\otimes_\zet}\, R$ has
$R$ as a module, with action $(a\oti b)\,.\,m \eq a\, b\, m$ for $a,b,m\iN R$.
Commutativity of $R$ ensures that this prescription constitutes an action, 
i.e.\ the equality of 
$ [(a_1\oti a_2)\, (b_1\oti b_2)]\,.\,m \eq a_1\,b_1\,a_2\,b_2\, m$
and $ (a_1\oti a_2)\,.\,[ (b_1\oti b_2)\,.\,m] \eq a_1\,a_2\,b_1\,b_2\, m$.
The categorification of commutativity of $R$ is the structure of a
braiding; accordingly we take the natural isomorphisms 
  \be
  \bearl
  \psi_{X,X',M}^{}: \quad (X\oti X')\mtimes M = U \oti U' \oti V \oti V' \oti M
  \\[3pt] \hspace*{12.8 em}
  \stackrel\cong\longrightarrow \,
  X\mtimes (X'\mti M) = U \oti V \oti U' \oti V' \oti M
  \eear
  \ee
(with $X \eq U\boti V \iN \cald$, $X' \eq U' \boti V' \iN\cald$ and $M\iN\calc$)
as the associativity constraints for the mixed tensor product of \P\ that are 
part of the defining data of the module category to be given by the braiding in 
\calc, i.e.
  \be
  \psi_{X,X',M}^{} = \id_U^{}\otimes c_{U',V}^{} \otimes \id_{V'\otimes M}^{} \,.
  \ee
This mixed associativity constraint has been derived geometrically from 
the theory of covering surfaces \cite[Eq.\,(17)]{barSc}. There is in fact a 
whole family of possible constraints for the mixed tensor product,
involving higher powers of the braiding $c$
\cite[Thm.\,2.2]{bfrs2}, but they are all equivalent \cite[Thm.\,2.4]{bfrs2}.

\paragraph{{\boldmath \P} as part of an equivariant topological field theory.}
 
That \P\ is a twist defect \cite{bajQ2} has the following mathematical 
formalization: the module category \P\ with underlying abelian category \calc\ 
is part of a more comprehensive structure \cite{barm2,barSc} -- together with 
the modular category  \cald\ it forms a $\zet_2$-\emph{equivariant modular 
category} \cite{TUra2,kirI14}. This 
amounts to the existence of further mixed fusion functors, including functors 
  \be
  \calc \times \cald \to \calc \qquad\text{and}\qquad
  \calc\times \calc\to\cald 
  \ee
and constraints which we will need later in our discussion.
The structure of a $\zet_2$-equivariant modular
category can be derived from a $\zet_2$-equivariant topological field theory
  \be
  \tft_\cald^{\zet_2}: \quad \cobord_{3,2,1}^{\zet_2} \longrightarrow \VVect
  \ee
which has as domain a cobordism category of manifolds with $\zet_2$-covers. The
functor $\tft_\cald^{\zet_2}$ can be given explicitly by applying the functor 
$\tft_\calc$ (see \erf{ord:cft}) for the (non-equivariant) topological field 
theory based on \calc\ to the total spaces of the covers,
  \be
  \tft_\cald^{\zet_2}(-) = \tft_\calc(\cov(-)) \,.
  \labl{tftz2}
For a detailed 
discussion in the framework of modular functors we refer to \cite{barSc}.

\paragraph{The Azumaya algebra {\boldmath \AP}.}
 
We need to collect a few facts about the module category \P. As any semisimple 
indecomposable (left) module category over a fusion category \cald, the 
category \P\ describing the permutation twist defect can be realized \cite{ostr}
as the category of (right) modules over an algebra \AP\ internal in \cald,
  \be
  \P \simeq \text{mod-}\AP \,.
  \ee
One possible choice for this algebra is the internal end of the tensor unit 
$\one\iN\calc$, i.e.\ $\AP \eq \underline\End_\calc(\one)$. As an object of the 
category \cald\ this is \cite{bfrs2}
  \be
  \AP = \bigoplus_{i \in I_\calc } \, U_i^\vee\boxtimes U_i \,,
  \labl{AP}
where the sum is over the isomorphism classes of simple objects of \calc;
its algebra structure is given explicitly in \cite[Thm.\,5.1]{bfrs2}.
The algebra \AP\ has a natural Frobenius algebra structure, which is 
presented in \cite[Prop.\,6.1]{bfrs2}. Furthermore, it is an 
\emph{Azumaya algebra} in \cald\ \cite[Thm.\,7.3]{bfrs2}. 
Let us explain the latter notion. For an algebra $A$ in a braided fusion 
category \C\ there are two \emph{braided induction functors}
  \be
  \alpha^\pm_{\!A}:\quad \calc\to A\bimod_\calc \,.
  \labl{+-}
They associate to an object $U\iN\calc$ the bimodule with underlying object
$A\oti U$ and with the left action $\rho^\pm \,{:=}\, m_A\oti\id_U\colon A
\oti A\oti U\To A\oti U$ given by multiplication $m_A$ in $A$, while the right 
action is $\ohr^+ \,{:=}\,(m_A\oti \id_U)\cir(\id_A\oti c^{}_{U,A})$ and
$\ohr^- \,{:=}\,(m_A\oti \id_U)\cir(\id_A\oti c_{A,U}^{-1})$, respectively.
The functors \erf{+-} have a natural structure of a monoidal functor; 
$A$ is called an Azumaya algebra iff $\alpha_{\!A}^+$, or equivalently
$\alpha_{\!A}^-$, is a monoidal equivalence. For \calc\ the monoidal category
of modules over a commutative ring, this coincides with the textbook
definition of Azumaya algebras \cite{vazh}.

\begin{rem}
(i)\, Recall that for describing a bilayer system we have to take the same 
braiding on the two copies of \calc\ in $\cald\eq \calc\boti\calc$.
There is another important structure of a braided fusion category on
the tensor category $\calc\boti\calc$, namely the `enveloping category'
$\calc^\env \eq \calc\boti\calc^\rev$ in which the second copy of \calc\ 
is instead endowed with the inverse braiding. If \calc\ is modular, then the 
enveloping category is a modular category as well; in fact, modularity implies 
that it is equivalent to the Drinfeld center of \calc, i.e.\ $\calc^\env
\,{\simeq}\, \Z(\calc)$. This structure of a modular tensor category is 
\emph{not} the one relevant for bilayer systems.
\\[2pt]
(ii)\, When regarded as an object of the enveloping category 
$\calc^\env \eq \calc\boti\calc^\rev$ the object \erf{AP} of the abelian 
category $\calc\boti\calc$ has again a natural Frobenius algebra structure, 
which is of interest in various other contexts, see e.g.\ 
\cite{muge8,ffrs,koRu,fuSs6}. This algebra structure on the object \erf{AP} is 
commutative with respect to the braiding of $\calc^\env$, rather than Azumaya.
\end{rem}

\paragraph{Braided induction for tensor products of algebras.}

Given two unital associative algebras $A_1$ and $A_2$ in a braided monoidal 
category, their tensor product $A_1\oti A_2$ can be endowed with the structure
of a unital associative algebra with multiplication
  \be
  (A_1\oti A_2)\otimes (A_1\oti A_2)
  \xrightarrow{~\idsm^{}_{A_1}\otimes c_{A_2,A_1}^{}\otimes\idsm^{}_{A_2}~}
  A_1\oti A_1 \oti A_2\oti A_2 \xrightarrow{~m_{A_1}^{}\otimes m_{A_2}^{}~}
   A_1\oti A_2 \,. 
  \ee
(Replacing the over-braiding $c_{A_2,A_1}^{}$ by the under-braiding 
$c_{A_1,A_2}^{-1}$ yields a different algebra structure on the same object 
$A_1\oti A_2$. This algebra is isomorphic, as an associative algebra, to the 
algebra structure on $A_2\oti A_1$ obtained with the convention chosen here.)

We will now establish a relation between the braided induction functors for the 
algebras $A_1$ and $A_2$ and those for $A_1\oti A_2$. We first introduce a 
functor
  \be
  \beta^+: \quad  A_2\bimod\ucalc \longrightarrow (A_1{\otimes} A_2)\bimod\ucalc
  \ee
that sends $B \,{\equiv}\, (B,\rho,\ohr) \iN A_2\bimod\ucalc$ to $A_1\oti B$ 
with the $(A_1{\otimes}A_2)$-bimodule structure given by the left action
  \be
  (A_1\oti A_2)\otimes (A_1\oti B)
  \xrightarrow{~\idsm_{A_1}\otimes c_{A_2,A_1}^{}\otimes\idsm_B~}
  A_1\oti A_1\oti A_2\oti B \xrightarrow{~m_{A_1}^{}\otimes \rho~} A_1\oti B
  \ee
and the right action
  \be
  (A_1\oti B)\otimes (A_1\oti A_2)
  \xrightarrow{~\idsm_{A_1}\otimes c_{B,A_1}^{}\otimes \idsm_{A_2}~}
  A_1\oti A_1\oti B\oti A_2 \xrightarrow{~m_{A_1}^{}\otimes \ohr~} A_1\oti B \,.
  \ee
Again, $\beta^+$ has a natural monoidal structure.  Moreover, one verifies that
  \be
  \alpha_{\!A_1\otimes A_2}^+ = \beta^+ \circ \alpha_{\!A_2}^+
  \ee
as monoidal functors. Similarly,we introduce another monoidal functor 
  \be
  \beta^-: \quad A_1\bimod\ucalc \rightarrow (A_1{\otimes} A_2)\bimod\ucalc \,,
  \ee
sending $B\,{\equiv}\, (B,\rho,\ohr) \iN A_1\bimod\ucalc$ to $A_2\oti B$ with 
the $(A_1{\otimes}A_2)$-bimodule structure given by the left action
  \be
  (A_1\oti A_2)\otimes (A_2\oti B) 
  \xrightarrow{~\idsm_{A_1}\otimes m_{A_2}^{}\otimes \idsm_B\,}
  A_1\oti A_2\oti B \xrightarrow{~c_{A_2,A_1}^{-1}\otimes \idsm_B\,}
  A_2\oti A_1\oti B \xrightarrow{~\idsm_{A_2}\otimes\rho\,} A_2\oti B 
  \ee
and the right action
  \be
  (A_2\oti B)\otimes (A_1\oti A_2)
  \xrightarrow{~\idsm_{A_2}\otimes\ohr\otimes\idsm_{A_2}\,}
  A_2\oti B \oti A_2 \xrightarrow{~\idsm_{A_2}\otimes c_{A_2,B}^{-1}\,}
  A_2\oti A_2\oti B \xrightarrow{~m_{A_2}^{}\otimes \idsm_B\,} A_2\oti B\,.
  \ee
By direct calculation one sees that the family 
  \be
  \nu_U^{} := c_{\!A_2,A_1}^{}\oti\id_U :\quad A_2\oti A_1\oti U 
  \longrightarrow A_1\oti A_2\oti U
  \ee
of isomorphisms, for $U\iN\calc$, furnishes a monoidal natural isomorphism 
$\nu\colon \beta^-\cir \alpha_{\!A_1}^- \,{\Longrightarrow}\,
\alpha_{\!A_1\otimes A_2}^-$. Similarly one verifies that the 
the same family of isomorphisms gives a monoidal natural isomorphism
  \be
  \tilde\nu:\quad \beta^-\circ \alpha_{\!A_1}^+ \Longrightarrow
  \beta^+ \circ \alpha_{\!A_2}^- \,.
  \ee
This is summarized in the

\begin{prop}
The following diagram of monoidal functors and monoidal natural isomorphisms
commutes:
  \be
  \xymatrix @R+6pt{
  && ~~~(A_1{\otimes}A_2)\bimod\ucalc~~~
  && \\ {}\save[]*\txt{%
    \begin{picture}(0,0)
    \put(54,-7)  {\begin{turn}{118}\xyRightarroww\end{turn}}
    \put(89,32)  {$\scriptstyle \nu$}
    \put(280,38) {\begin{turn}{62}\xyRightarroww\end{turn}}
    \put(332,24) {$\Idsm$}
    \end{picture}}%
  \restore
  &  ~~~~~A_1\bimod\ucalc~~~ \ar_{\beta^-}[ur] \ar@{=>}^{\tilde\nu}[rr]
  && ~~~A_2\bimod\ucalc~~~~~ \ar^{\beta^+}[ul]
  & \\
  \calc \ar_{\alpha_{\!A_1}^-}[ur] \ar@/^5pc/[uurr]^{\alpha_{\!A_1\otimes A_2}^-}="2"
  && \calc \ar_{\alpha_{\!A_2}^-}[ur]\ar^{\alpha_{\!A_1}^+}[ul]
  && \calc \ar^{\alpha_{\!A_2}^+}[ul] \ar@/_5pc/[uull]_{\alpha_{\!A_1\otimes A_2}^+}="1"
  }
  \labl{bigdiagram}
\end{prop}

\paragraph{The Azumaya algebra {\boldmath $\AP\oti\AP$}.}

We use these observations to study the algebra $\AP\oti \AP$ internal in \cald. 
As a tensor product of two Azumaya algebras, it is again Azumaya. 
Now recall \cite[Cor.\,3.8]{danO} that, up to equivalence, an
indecomposable module category \calm\ over a modular tensor category
\cald\ is characterized by a pair $B_1,B_2$ of connected \'etale algebras in 
\cald\ together with a braided equivalence $\Psi_{\!\calm}\colon
B_1\Mod^0\ucald \,{\stackrel\simeq\longrightarrow}\, B_2\Mod^{0\;\text{rev}}\ucald$ 
between the category of local $B_1$-mo\-dules and the reverse of the category of 
local $B_2$-mo\-dules. It follows from the results of \cite{ffrs} that 
for the module category $\calm \eq$mod-$A$ of right modules over an 
algebra $A\iN\cald$ these characteristic data can be extracted from the 
braided induction functors $\alpha^\pm_{\!A}\colon \cald\To A$-bimod$\ucald$. 
In the particular case that $A$ is an Azumaya algebra, the two functors
$\alpha_{\!A}^\pm$ are monoidal equivalences and the two \'etale algebras 
$B_1$ and $B_2$ are just the tensor unit $\one$, so that 
$B_1\Mod\ucald^0 \eq \cald \eq B_2\Mod\ucald^0$. Moreover, in this case the 
braided equivalence $\Psi_{\!\text{mod-}A}\colon \cald\To\cald$ is given by
  \be
  \Psi_{\!\text{mod-}A} = (\alpha_{\!A}^+)^{-1}_{} \circ\alpha_{\!A}^- \,.
  \labl{PsiA}
Now according to \cite[Prop.\,7.3]{bfrs2}, for the Azumaya algebra \AP\ the
local induction functors satisfy
$\alpha_\AP^+(U \boti V) \,{\cong}\, \alpha_\AP^-(V \boti U),$. This implies 
that the functor $\Psi_{\text{mod-}\AP}$ acts on objects and morphisms by 
permutation, in particular 
  \be
  \Psi_{\!\text{mod-}\AP}(U\boti V) = V\boti U \,.
  \labl{PsiAp}

\begin{rem}
We can now see that the module category \P\ over \cald\ describes the 
permutation twist surface defect of 
\cite{bajQ4}. To this end  we note \cite[Sect.\,4]{fusV} that for a defect 
surface labeled by a \cald-module \calm\ the functor $\Psi_{\!\calm}$ 
describes the transmission of bulk Wilson lines through the defect surface.
Thus for $A \eq \AP$ bulk Wilson lines get permuted according to \erf{PsiAp} 
when passing through the defect surface. This is precisely the property
characterizing the permutation twist defect \cite{bajQ4}.
\end{rem}

\begin{lem}
Let $A$ and $A'$ be  Azumaya algebras in a braided fusion category. Then we 
have the isomorphism 
  \be
  \Psi_{\!\text{\rm mod-}A\otimes A'} \,\cong\,
  \Psi_{\!\text{\rm mod-}A} \circ \Psi_{\!\text{\rm mod-}A'}
  \labl{PsiAA'}
of monoidal functors.
\end{lem}

\begin{proof}
Given the definition \erf{PsiA} of the monoidal functors, the statement
follows immediately from the commuting diagram \erf{bigdiagram},
from which one can also read off the monoidal natural isomorphism.
\end{proof}

\begin{prop}\label{prop1}
The Azumaya algebra $\AP\oti \AP$ in \cald\ is Morita equivalent to the
tensor unit $\oned$ of \cald.
\end{prop}

\begin{proof}
Combining \erf{PsiAp} and \erf{PsiAA'} we have $\Psi_{\!\text{mod-}\AP\otimes\AP}
\,{\cong}\, \Psi_{\!\text{mod-}\AP} \cir \Psi_{\!\text{mod-}\AP} \eq \Id_\cald$.
Thus by \cite[Cor.\,3.8]{danO} the module categories mod-$\AP\Oti\AP$ and 
mod-$\oned$ are equivalent, i.e.\ $\AP\oti \AP$ and $\oned$ are 
Morita equivalent.
\end{proof}

\paragraph{Categories of defect Wilson lines.}
 
We are now in a position to find the relevant categories of defect Wilson 
lines. Consider two types of surface defects separating a topological phase
of type \cald\ from itself, corresponding to two \cald-modules $\calm_1$ 
and $\calm_2$. According to \cite{fusV} the category of surface Wilson lines 
separating $\calm_1$ from $\calm_2$ is the functor category 
$\Fun_\cald(\calm_1,\calm_2)$ of $\cald$-module functors. If the left module 
categories $\calm_1$ and $\calm_2$ are realized as the categories of right 
modules over algebras $A_1$ and $A_2$ in \cald, respectively, then this 
functor category is equivalent to the category of $A_1^{}\oti A_2\op$-modules
in \cald. We are interested in the situation that the two \cald-modules in
question are either $\T_\cald$ or \P, corresponding to the transparent and to 
the permutation twist defect. Proposition \ref{prop1} tells us that under 
forming tensor products the Azumaya algebra \AP\ for the twist defect \P\ has 
order two up to Morita equivalence; as a consequence, when calculating the 
functor categories we can work with \AP\ in place of $A_{\mathcal P}\op$. We 
then find the following categories of defect Wilson lines:
  \def\leftmargini{1.27em}~\\[-1.25em]\begin{itemize}\addtolength{\itemsep}{-7pt}

\item
The category $\Fun_\cald(\T_\cald,\T_\cald)$ of defect Wilson
lines separating the transparent surface defect from itself is just
\cald, as expected:
  \be
  \Fun_\cald(\T_\cald,\T_\cald) \simeq (\oned{\otimes}\oned)\Mod\ucald
  \cong \oned\Mod\ucald \cong \cald \,.
  \labl{funTT}

\item
There are two categories of defect Wilson lines separating the transparent
defect from the twist defect, $\Fun_\cald(\T_\cald,\P)$ and 
$\Fun_\cald(\P,\T_\cald)$; we find
  \be
  \Fun_\cald(\T_\cald,\P) \simeq (\oned{\otimes}\AP)\Mod\ucald
  \cong \AP\Mod\ucald \cong \calc 
  \ee
and, in a similar manner, $\Fun_\cald(\P,\T_\cald) \,{\cong}\, \calc$. As shown 
in \cite[Sect.\,6.2]{fusV}, each such Wilson line labeled by $W\iN\calc$ gives 
rise to a (special symmetric Frobenius) algebra in the Morita class of \AP; this
is actually just the internal end $\underline\End_\calc(W)$ 
\cite[Sect.\,3]{ostr}.

\item
Finally, the category of defect Wilson
lines separating the twist defect from itself is
  \be
  \Fun_\cald(\P,\P) \simeq (\AP{\otimes}\AP)\Mod\ucald
  \simeq \oned\Mod\ucald \cong \cald \,.
  \labl{funPP}
\end{itemize}
The category describing Wilson lines which separate the twist defect from the 
transparent defect provides the labels for a permutation-type ``genon'' 
of \cite{bajQ2}. Genons are thus labeled by objects of the category \calc.

\paragraph{Fusion of surface defects.}
 
Topological surface defects can be fused. In the particular situation of two
surface defects separating a topological phase of type \calc\ from itself, 
described by \calc-mo\-du\-les $\calm_1$ and $\calm_2$, the fusion product is 
the \calc-module $\calm_1 \,{\boxtimes_\calc}\, \calm_2$. The Deligne product 
$\boxtimes_\calc$ of two module categories over a braided fusion category is 
the categorification of the tensor product $M_1 \,{\otimes_R}\, M_2$ of two 
left modules $M_1$ and $M_2$ over a commutative ring $R$ and has a similar 
universal property.  For a precise definition see 
\cite[Def.\,3.3\,\&\,Sect.\,4.4]{enoM}. 
By \cite[Prop.\,3.5]{enoM} there is an equivalence
  \be
  \calm_1^{} \boxtimes_\calc \calm_2^{} \simeq \Fun_\calc(\calm_1\op,\calm_2^{}) 
  \labl{MoDM}
of abelian categories. 

The right hand side of \erf{MoDM} has a natural interpretation \cite{fusV} as 
the category of surface Wilson lines that separate the surface defect labeled 
by $\calm_1\op$ from the one labeled by $\calm_2^{}$.
This is no coincidence: if we realize two module categories $\calm_1$ and 
$\calm_2$ over a modular tensor category \calc\ as the categories of right 
modules over algebras $A_1$ and $A_2$ internal in \calc, the fused module 
category $\calm_1 \,{\boxtimes_\calc}\, \calm_2$ is realized by the category of 
right $A_1{\otimes}A_2$-modules. Now consider a defect surface with the topology 
of a plane, separated by a Wilson line into two half planes labeled by $A_1$ 
and by $A_2$, respectively. Wilson lines of this type are labeled by the 
category of $A_1$-$A_2$-bimodules which equals, as an abelian category,
$(A_1^{}{\otimes}A_2\op)\Mod\ucalc$. By folding 
the plane along the Wilson line we arrive at a configuration in which a Wilson 
line separates the transparent surface defect $\T_\calc$ from the surface 
defect that is obtained by fusing the surface defect with label $A_1$ with the
orientation-reversed surface defect for $A_2$. This is the \calc-module category
  \be
   A_1^{}\Mod\ucalc \,\boxtimes_\calc A\op_2\Mod\ucalc 
  \cong (A_1^{}{\otimes} A_2\op)\Mod\ucalc \,.
  \ee
Now for any \calc-module category \calm, the category of surface Wilson lines 
separating the transparent defect $\T_\calc$ from \calm\ is 
$\Fun_\calc(\calc,\calm) \,{\cong}\, \calm$, and hence in the case at hand the 
abelian category $(A_1^{}{\otimes} A_2\op)\Mod$. Thus the equality of the two 
abelian categories can be understood through the folding procedure and
provides a consistency check on the description of the fusion of surface defects
by the Deligne product.

It follows in particular that the transparent defect, described by the algebra 
$\one_\calc$ in \calc, acts as a (bi)monoidal unit. In the case of permutation 
twist defects in the bilayer system based on $\cald=\calc\boxtimes\calc$, we get
  \be
  \P \boxtimes_\cald P_\cald \cong (\AP {\otimes} \AP)\Mod\ucald \cong \cald \,,
  \ee
where in the last step we used again that the Azumaya algebra $\AP\oti \AP$ is 
Morita equivalent to the tensor unit. Thus the fusion product of the twist defect
\P\ with itself is the transparent defect; this is certainly not unexpected.

\paragraph{More general Wilson lines.}
 
A topological field theory of Reshetikhin-Turaev type actually admits
more general types of Wilson lines, in which any finite number of surface 
defects meet. Locally in a three-manifold, the situation looks like in the 
following picture, in which, as in the formalism used in \cite{fusV2}, the 
locus of the Wilson line in a three-manifold is actually a tube:
  \eqpic{Figure13+Figure14}{410}{44} {
  \put(145,-2) {\Includepicfj{22}{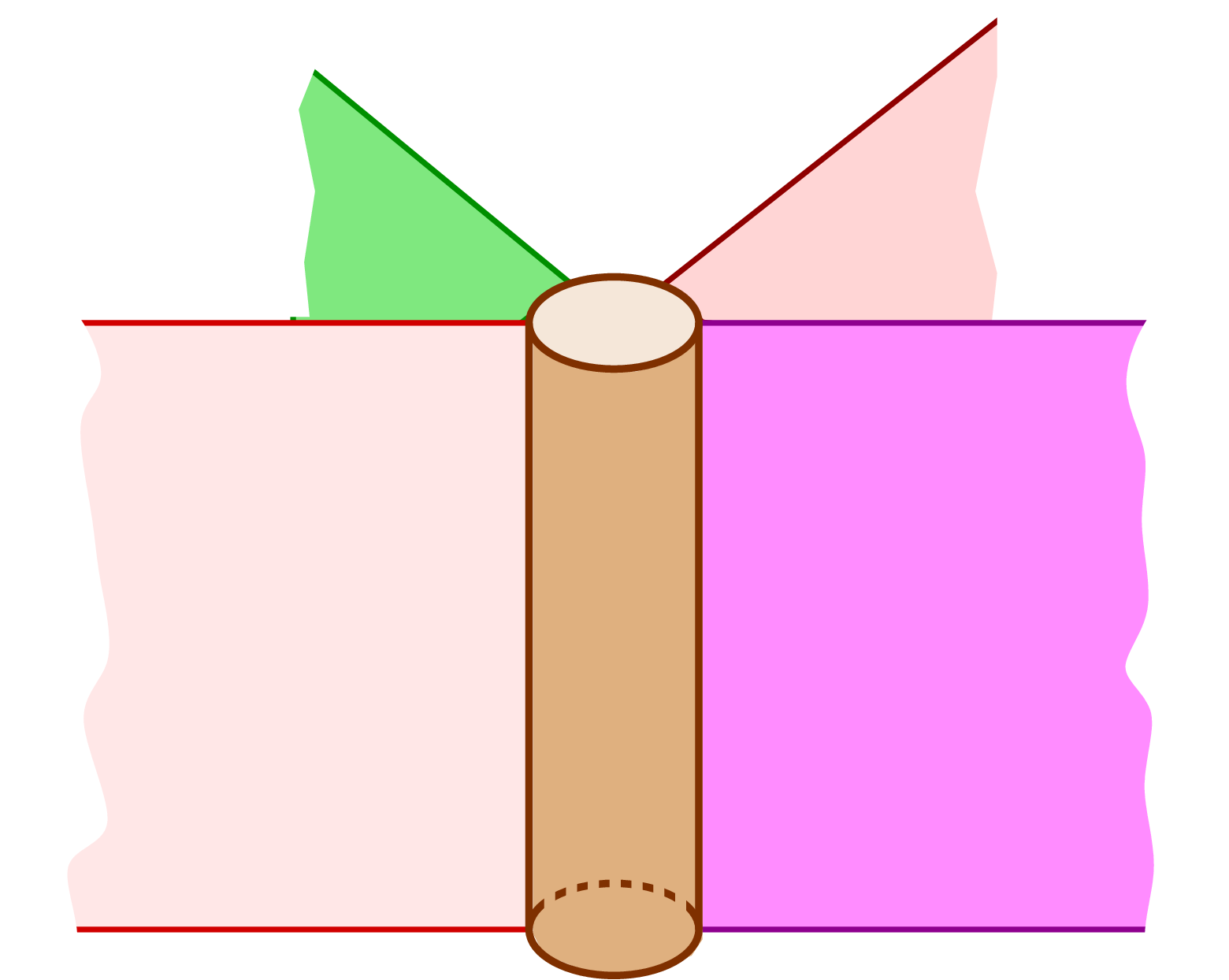}
  }
  }
In situation at hand, each of the surface defects that meet at the Wilson line
can be either the transparent defect or the twist defect. We denote the 
decorated one-dimensional manifold consisting of a circle with an 
$n_\T$-tuple $\vec p \,\,{\equiv}\, (p_1,p_2,...\,,p_{n_\T)}$ of 
points marked with the transparent defect $\T_\cald$ and an $n_\Po$-tuple
$\vec q \,\,{\equiv}\, (q_1,q_2,...\,,q_{n_\Po})$ of points marked by the twist 
defect \P\ by \Spq, and the associated category of generalized Wilson lines 
by $\tft(\Spq)$. Because of the $\zet_2$-fusion rules obeyed by the transparent 
and twist defects, the category $\tft(\Spq)$ is equivalent to \cald\ if 
$n_\Po$ is even, and equivalent to \calc\ if $n_\Po$ is odd.

A geometric understanding of the categories of generalized Wilson lines
is provided by the \emph{cover functor} $\cov$ of \cite[Prop.\,2]{barSc}, 
which maps the decorated one-manifold \Spq\ as follows to a two-sheeted cover 
of the circle $S^1$. First, take the disjoint union of two copies 
$\tilde S^{(1)}$ and $\tilde S^{(2)}$ of the non-connected manifold obtained by 
replacing the open intervals in $S^1\,{\setminus}\,(\vec p\,{\cup}\,\vec q)$ by 
closed intervals.  For each marked point $q_i^{}$ this gives two points 
$q_i^{1,l}$ and $q_i^{1,r}$ on $\tilde S^{(1)}$ and two points $q_i^{2,l}$ and 
$q_i^{2,r}$ on $\tilde S^{(2)}$. They are associated to the interval on the left 
hand side and to the one on the right hand side of $q_i \iN S^1$, respectively.
Next we identify $q_i^{2,l}$ with $q_i^{1,r}$ and $q_i^{1,l}$ with $q_i^{2,r}$.
For the points $p_i$ a similar construction is performed, but this time we
identify $p_i^{1,l}$ with $p_i^{,1,r}$ and $p_i^{2,l}$ with $p_i^{2,r}$. 
The two different identifications are illustrated in the left and right hand
parts of the following figure:
  \eqpic{figure1}{280}{30} {
  \put(0,-2) {\Includepicfj{62}{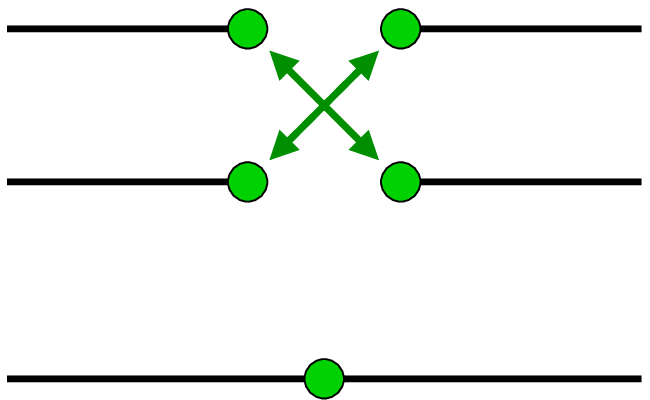}
  \put(-25,49.4) {$\cov(\Sigma)$}
  \put(-12,-.2)  {$\Sigma$}
  \put(37.5,25.5){\sse$ q^{1,l} $}
  \put(37.5,71.5){\sse$ q^{2,l} $}
  \put(55.5,-6.1){\sse$ q $}
  \put(64.5,25.5){\sse$ q^{1,r} $}
  \put(64.5,71.5){\sse$ q^{2,r} $}
  }
  \put(180,-2) {\Includepicfj{62}{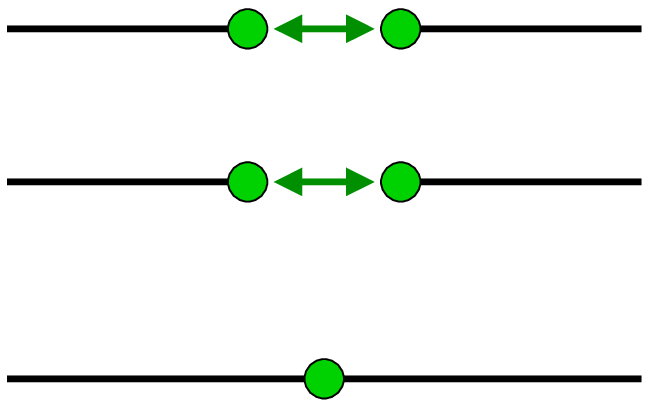}
  \put(-25,49.4) {$\cov(\Sigma)$}
  \put(-12,-.2)  {$\Sigma$}
  \put(37.5,25.5){\sse$ q^{1,l} $}
  \put(37.5,71.5){\sse$ q^{2,l} $}
  \put(55.5,-6.1){\sse$ q $}
  \put(64.5,25.5){\sse$ q^{1,r} $}
  \put(64.5,71.5){\sse$ q^{2,r} $}
  } }
We have thus associated a two-sheeted cover $\cov(S) \To S$ to the decorated 
one-manifold \Spq. As already indicated in \erf{tftz2}, the equivariant 
topological field theory $\tft_\cald^{\zet_2}$ is obtained \cite{barSc} by 
applying the TFT functor associated to \calc\ to the cover.

Let us check that the categories of generalized Wilson lines for decorated
one-manifolds that we have computed in \erf{funTT}\,--\,\erf{funPP}
coincide with the evaluation of the 2-functor $\tft_\calc$ on the total
space $\cov(S)$ of the two-sheeted cover of $S$. If
the number $n_\Po$ of twist defects is even, then $\cov(S) \To S$ is the 
trivial cover whose total space has two connected components; we thus get
  \be
  \tft^{\zet_2}_\cald(S)
  = \tft_\calc(S^1\,{\sqcup}\,S^1) \simeq \tft_\calc(S^1)\boxtimes\tft_\calc(S^1)
  = \calc \boxtimes \calc = \cald \,.
  \ee
If $n_\Po$ is odd, then the total space $\cov(S)$ is connected and we obtain
instead the category $\tft^{\zet_2}_\cald(S) \eq \tft_\calc(S^1) \eq \calc$.


\section{Generalized conformal blocks and their dimensions}
\label{sec:3}

To an oriented surface $\Sigma$ with boundaries, an extended topological field 
theory assigns a functor. More explicitly, given a decomposition 
$\partial\Sigma \eq {-}{\partial\Sigma_-} \,{\sqcup}\, \partial\Sigma_+$
of the boundary into incoming and outgoing parts, we get a functor
  \be
  \tft_\cald(\Sigma): \quad \tft_\cald(\partial \Sigma_-) \longrightarrow
  \tft_\cald(\partial \Sigma_+) \,.
  \ee
To achieve a detailed understanding of these functors, two particular 
perspectives prove to be helpful: First, the functors for
particularly simple surfaces can be assembled to obtain those for more
complicated surfaces. Second, by assigning specific objects of \cald\ to
the boundary surfaces, one arrives at vector spaces of (generalized) 
conformal blocks. We address both of these points of view.

\paragraph{Mixed tensor products.}

The basic observation that allows for the first perspective is that any 
oriented surface with boundary admits a decomposition into pairs of pants,
cylinders and disks. Evaluating the TFT-functor on the pair of pants $\pop$ 
with two incoming  and one outgoing boundary circles gives a functor
  \be
  \tft_\cald(S^1) \boxtimes \tft_\cald(S^1) \cong
  \tft_\cald(S^1 \,{\sqcup}\, S^1) \longrightarrow \tft_\cald(S^1) \,.
  \ee
This endows the category \cald\ associated to the circle with a tensor product 
functor. (An associativity constraint for this tensor product is then
provided by the natural transformation that the TFT associates to a suitable
three-manifold with corners; this way \cald\ becomes a monoidal category.)

Once we allow for non-trivial defects, we get additional categories 
associated to decorated circles, and thereby additional  tensor products; 
they relate different categories and are thus \emph{mixed} tensor products. 
In the case of twist defects, these tensor products can be extracted from 
the underlying equivariant topological field theory. If we deal with
a permutation equivariant theory, mixed tensor products have 
been computed \cite[Sect.\,4.3]{barSc} with the help of the cover functor, 
which we already encountered in \erf{tftz2}. Proceeding in this way we find:
  \def\leftmargini{1.27em}~\\[-1.25em]\begin{itemize}\addtolength{\itemsep}{-7pt}

\item
Denote by $n_1$ and $n_2$ the numbers of twist defects ending on the two ingoing 
circles and by $n_3$ the number of twist defects of the outgoing circle. On the 
pair of pants we then have $(n_1 {+} n_2 {+} n_3) /2$ lines that connect 
boundary circles, all labeled by the twist defect. The following picture shows 
a case with $n_1 \eq 5$, $n_2 \eq 7$ and $n_3 \eq 6$ and with each line 
connecting two different circles:
  \eqpic{figure3}{130}{42} {
  \put(0,-9) {\Includepicfj{35}{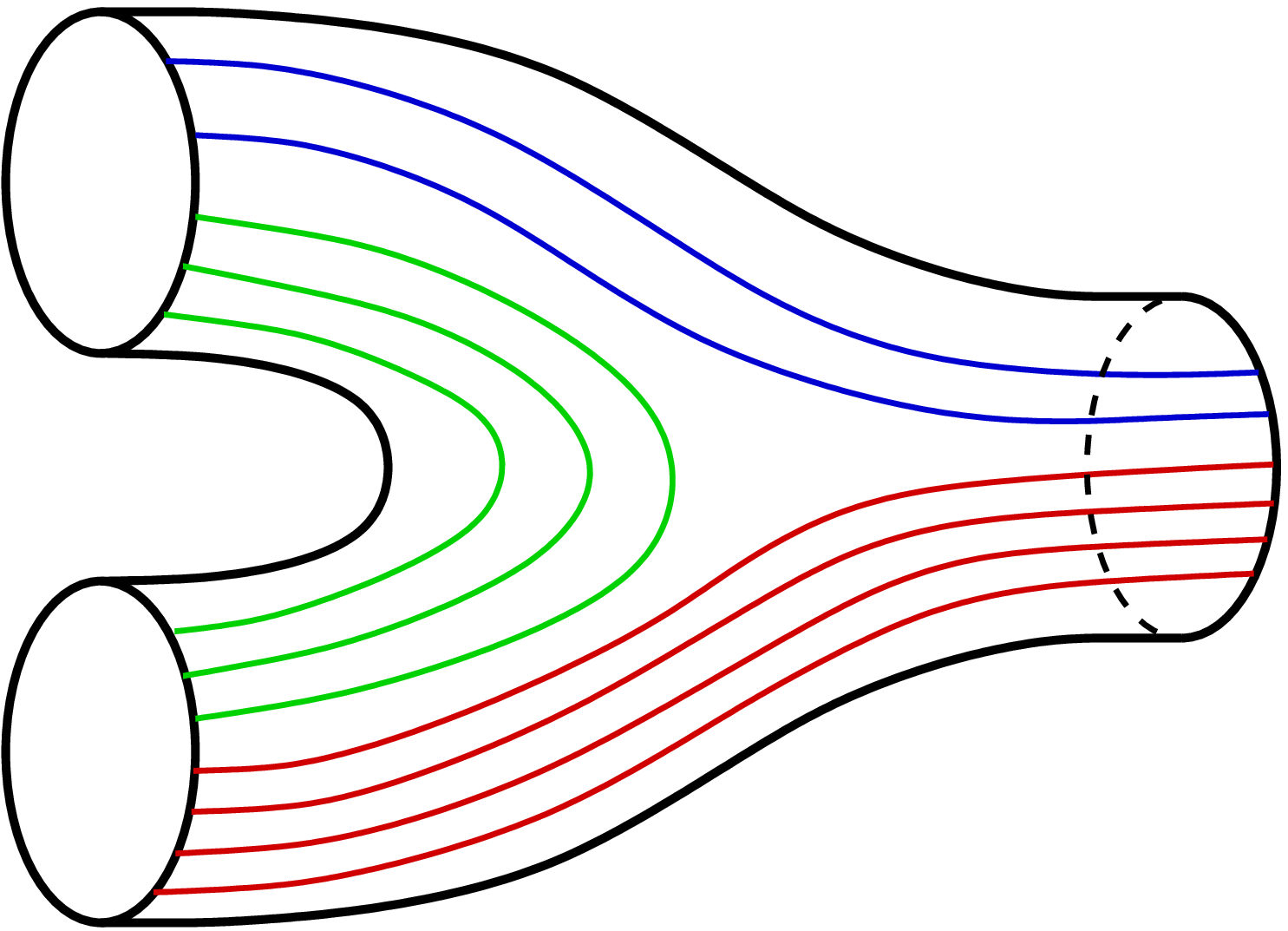}
  } }

\item
If the numbers $n_1$ and $n_2$ are both even, then $n_3$ is necessarily even as 
well. Thus the categories assigned to the boundary circles are all equivalent 
to \cald. According to \cite{barSc}, the functor associated to any such pair 
of pants is the tensor product in the monoidal category \cald:
  \be
  \begin{array}{rrll}
  \tft_\cald(\pop_{n_1,n_2,n_3}) :& \cald\times\cald &\!\longrightarrow\! &\cald
  \\[3pt]
  & (U_1{\boxtimes} U_2) \Times (V_1{\boxtimes} V_2) &\! \longmapsto \!&
  (U_1{\otimes}V_1) \boti (U_2{\otimes}V_2) \,.
  \eear
  \ee

\item
If $n_1$ is even and $n_2$ is odd, then $n_3$ is necessarily odd, and the 
functor is the mixed tensor product:
  \be
  \begin{array}{rrll}
  \tft_\cald(\pop_{n_1,n_2,n_3}) : & \cald\times\calc &\!\longrightarrow\! &\calc
  \\[3pt]
  & (U_1{\boxtimes} U_2) \Times M &\! \longmapsto \!& U_1\oti U_2\oti M \,,
  \eear
  \ee
where on the right hand side the tensor product in \calc\ appears.
\\
The situation is analogous when $n_1$ is odd and $n_2$ is even.

\item
If both $n_1$ and $n_2$ are odd, then $n_3$ is even. In this case
the functor is the one computed in \cite[Sect.\,4.3,\, p.\,314]{barSc}:
  \be
  \begin{array}{rrll}
  \tft_\cald(\pop_{n_1,n_2,n_3}) : & \calc\times\calc &\!\longrightarrow\!& \cald
  \\[3pt]
  & M \times N &\! \longmapsto \!&
  \bigoplus_{i\in I_\calc}\, (M\oti N \oti S_i^\vee) \boti S_i^{} \,.
  \eear
  \ee
\end{itemize}

\noindent
Again we have a geometric understanding of these functors. Consider a pair 
of pants $\pop$ with a pattern of non-intersecting surface defect lines, 
as in the figure \erf{figure3}. And again we take two copies 
of $\pop$ with the defect lines removed and glue them together according to
the prescription in the figure \erf{figure1}, with the identification
depending on whether the defect line is labeled by the transparent
defect or the twist defect. This way we get a two-sheeted cover 
$\tilde\pop _{n_1,n_2,n_3}\To \pop_{n_1,n_2,n_3}$ which restricts 
on the boundary circles to the cover constructed in Section \ref{sec:2}. 
The functors $\tft_\cald(\pop_{n_1,n_2,n_3})$ just described are then
the functors $\tft_\calc(\tilde \pop_{n_1,n_2,n_3})$.

\paragraph{Spaces of conformal blocks.}

We now turn to the second perspective and focus on spaces of conformal blocks.
Consider an oriented surface $\Sigma$, for the moment without defects, with 
only ingoing boundary circles, i.e.\
$\partial\Sigma \eq {-}\partial\Sigma_- \,{\cong}\, (S^1)^{\sqcup n}$ and
$\partial_+\Sigma \eq \emptyset$. The extended topological field theory provides
a functor $\tft_\cald(\Sigma)\colon \cald^{\boxtimes n}\To\Vect$. Specifying an 
object in \cald\ for each boundary circle, we obtain a vector space, known as 
a space of \emph{conformal blocks}. In applications to topological quantum 
computing, these spaces are the spaces of ground states and are thus the 
recipients of qubits. The dimension of the space of conformal blocks, and 
thus the ground state degeneracy, is computed by the Verlinde formula. 

In the presence of surface defects in the \threedim\ theory, the topological 
surface $\Sigma$ is endowed with a collection of  non-intersecting lines. Such 
lines are closed or have end points on the boundary circles of $\Sigma$. Each 
segment of a line is labeled either by the transparent defect $\T_\cald$ or by 
the twist defect \P. A typical situation is displayed in the following figure:
  \eqpic{figure2}{240}{35} {
  \put(0,-7) {\Includepicfj{35}{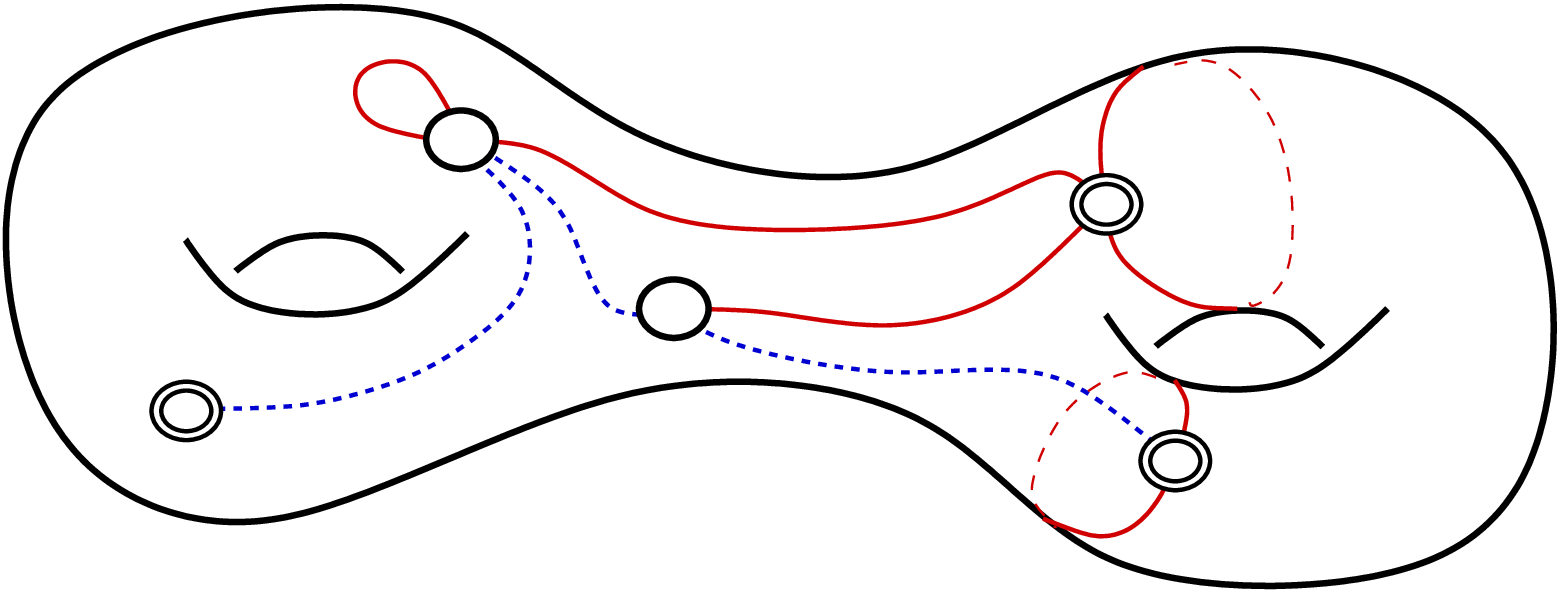}
  } }
Here lines labeled by the transparent defect $T_\cald$ are drawn as dotted 
lines, while those labeled by the twist defect \P\ are drawn as solid lines. 

A boundary circle is drawn as a double line if an even number of \P-lines ends 
on it; the associated category is $\cald \eq \calc\boti\calc$. 
Boundary circles with an odd number of \P-lines are drawn as single lines; 
the corresponding category is \calc. Among the single-line circles are 
those on which only one \P-line ends; these are ``genons''. (Thus in the picture 
\erf{figure2} there is one genon, the lower of the two single-line circles.)

Our task is to construct for a surface $\Sigma$ with $m_0$ boundary circles having
an even number of \P-lines and $m_1$ boundary circles having an odd number
a functor
  \be
  \tft_\cald(\Sigma):\quad \cald^{\boxtimes m_0}\boxtimes\calc^{\boxtimes m_1}
  \longrightarrow \Vect \,\,  
  \labl{covfunc}
that describes generalized conformal blocks, including their dependence on the 
labels of the boundary circles. By the axioms of topological field 
theory, the gluing of surfaces must give rise to the composition of the 
associated functors. As a consequence, the functor $\tft_\cald(\Sigma)$
can be expressed as a composite of the functors arising in a pair-of-pants 
decomposition of the surface $\Sigma$. The latter are already known  from the 
previous discussion: they are provided by the equivariant topological
field theory $\tft_\cald^{\zet_2}$. Note that surfaces with $\zet_2$-covers
can be glued, and in an equivariant topological field theory this
translates into the composition of functors.

Hereby we are led to the following generalization of the construction for 
pairs of pants given above.
To a surface $\Sigma$ with embedded $\T_\cald$-lines and \P-lines we associate
a two-sheeted cover $\cov(\Sigma) \To \Sigma$ as follows: Again we glue together 
two copies of $\Sigma$ with all defect lines removed, in a way that is 
determined by the surface defect labeling the defect line, as in the figure 
\erf{figure1}. (A variant is to glue standard disks to
the boundary circles of $\Sigma$ so as to get closed oriented surfaces. In this 
formulation one gets two-fold \emph{branched} covers, with branch points in 
disks whose boundaries contain an odd number of
twist defects. As mentioned, circles with one twist defect 
describe genons; they are thus end points of branch cuts, compare \cite{bajQ4}.)

It is crucial that the construction of $\zet_2$-covers is compatible with
the gluing of surfaces with defects. Indeed, consider the surface 
$\Sigma_1{\#}\Sigma_2$ that is obtained by gluing together, along appropriate 
boundary circles, surfaces $\Sigma_1$ and $\Sigma_2$ with defects. Then the 
$\zet_2$-cover of $\Sigma_1\#\Sigma_2$ furnished by our construction is the 
same as the surface obtained by gluing the cover $\cov(\Sigma_1) \To \Sigma_1$ 
to the cover $\cov(\Sigma_2) \To \Sigma_2$,
  \be
  \cov(\Sigma_1{\#}\Sigma_2) = \cov(\Sigma_1) \,\#\, \cov(\Sigma_2) \,.
  \ee
We thus conclude that the generalized conformal block functor for a general 
surface as in \erf{covfunc} 
is obtained by applying $\tft_\calc$ to the two-fold cover $\cov(\Sigma)$, i.e.\ 
we have $\tft_\cald(\Sigma) \eq \tft_\calc(\cov(\Sigma))$ as in \erf{tftz2}.
This provides a model independent confirmation of an insight gained
in the study \cite{bajQ2} of several classes of models.

\paragraph{A generalization of the Verlinde formula.}

What we have achieved is to identify the generalized conformal blocks
in the bilayer topological field theory based on \cald\ associated to a surface 
with defects with ordinary conformal blocks for $\tft_\calc$ on the cover
of that surface. As a consequence, we can compute the dimension of these spaces
with the ordinary Verlinde formula for the theory based on \calc. 

Consider a closed surface of genus $g$, and thus of Euler characteristic $\chi 
\eq 2\,{-}\,2g$, obtained by gluing disks to the boundary circles of a surface 
$\Sigma$. Let $\Sigma$ have $N_0$ boundary circles with an even number of twist 
defects, labeled with objects $U_i\boti\tilde U_i \iN\cald$ for 
$i\eq 1,2,...\,, N_0$, where $U_i,\tilde U_i\iN\calc$, and $N_1$ boundary 
circles with an odd number twist defects, labeled with objects $V_j\iN\calc$ 
for $j\eq 1,2,...\,, N_1$.  Then by the Riemann-Hurwitz theorem the cover 
$\cov(\Sigma)$ has Euler characteristic $2\chi \,{-}\,N_1$, i.e.\ the genus of 
the cover increases linearly with $N_1$. A boundary circle of $\Sigma$ with 
an even number of twist defects has a pre-image on $\cov(\Sigma)$ consisting
of two circles; we label them by the objects $U_i$ and $\tilde U_i$ of \calc,
respectively, with the relevant value of $i$. A boundary circle of $\Sigma$ with 
an odd number of twist defects has a single circle as its pre-image, which
we label by the appropriate object $V_j\iN\calc$. Taking for simplicity the 
objects $U_i$, $\tilde U_i$ and $V_j$ to be simple, we arrive this way at the 
following formula for the dimensions of spaces of generalized blocks:
  \be
  \dimc\big( \tft_\cald(\Sigma;\, \{ U_i{\boxtimes}\tilde U_i \}\,,\{ V_j \} )\big)
  = \sum_{n\in I_\calc} (S_{0,n})^{2\chi-N_1}
  \prod_{i=1}^{N_0} \frac{S_{U_i,n}}{S_{0,n}}\, \frac{S_{\tilde U_i,n}}{S_{0,n}}
  \prod_{j=1}^{N_1} \frac{S_{V_j,n}}{S_{0,n}} \,.
  \labl{generikverl}
Here $S$ is the modular S-matrix of the category \calc, i.e.\ the matrix 
obtained from $s \eq (s_{i,j})$ as given in \erf{smat} by rescaling such that 
$S$ is unitary and symmetric, and where $0\iN I_\calc$ is the isomorphism 
class of $\one_\calc$.

For instance, for $\Sigma \eq S^2$ of genus 0, $N_0 \eq 0$ and all $V_j$ equal,
the dimension is
  \be
  \dimc\big( \tft_\cald(S^2;\, \emptyset \,,\{ V,V,...\,,V \} )\big)
  = \sum_{n\in I_\calc} (S_{0,n})^{4-2N_1}_{} (S_{V,n})^{N_1}_{} .
  \labl{Vf00N}

\paragraph
{Dependence of the dimension of spaces of conformal blocks on the genon type.}

In the context of quantum computing twist defects are of interest because
the relevant spaces of conformal blocks are associated with surfaces of higher 
genus, so that they generically have larger dimension than conformal blocks for 
surfaces without twist defects. From this perspective it should also be 
appreciated that each genon comes with the choice of a label, which is an object
in the category \calc.  Since this datum enters in the dimension formula 
\erf{generikverl}, it constitutes an additional handle on increasing the 
dimension of the space of conformal blocks.

As a simple instructive example, take a sphere with four genons, i.e.\
$\Sigma \eq S^2$, $N_0 \eq 0$ and $N_1 \eq 4$, and take \calc\ to be the modular
tensor category of the critical Ising model, describing a free Majorana fermion.
(For a related discussion of genons in this model see \cite[Sect.\,V]{bajQ2}.)
Then $\cov(\Sigma) \eq \mathrm T$ is a torus with four boundary circles, and the 
set $I_\calc$ of isomorphism classes of simple objects of \calc\ has three 
elements, $I_\calc \eq \{ \one_\calc,\,\sigma,\,\epsilon \}$.
We consider two extreme choices for the labels of the genons. First, let
all genons be labeled by $\one_\calc$, which is the monoidal unit of \calc. 
Then the dimension of the space of generalized conformal blocks is
  \be
  \bearll
  \mathrm d_4(\one_\calc) \!\!\! &
  := \dimc\big( \tft_\cald( S^2;\, \emptyset
  \,,\{ \one_\calc,\one_\calc,\one_\calc,\one_\calc \} )\big)
  \Nxl2&
  \,= \dimc\big( \tft_\calc( \mathrm T; \one_\calc^{\;\otimes 4} )\big) 
  = \dimc\big( \tft_\calc( \mathrm T; \one_\calc )\big) = \big| I_\calc \big| = 3 \,.
  \eear \hsp{1.9}
  \ee
Second, if all genons are labeled by $\sigma$, then by using the
fusion rules $\sigma\oti\sigma \,{\cong}\, \one_\calc\,{\oplus}\, \epsilon$ 
and $\epsilon\oti\epsilon \,{\cong}\, \one_\calc$ we get
  \be
  \bearll
  \mathrm d_4(\sigma) \!\!\! & 
  := \dimc\big( \tft_\cald( S^2;\, \emptyset \,,\{ \sigma,\sigma,\sigma,\sigma \} )\big)
  \Nxl2&
  \,= \dimc\big( \tft_\calc( \mathrm T; \sigma^{\otimes 4} )\big)
  = \dimc\big( \tft_\calc( \mathrm T; 2\,\one_\calc\,{\oplus}\,2\,\epsilon \big)
  \Nxl2&
  \,= 2\, \big[ \dimc\big( \tft_\calc( \mathrm T; \one_\calc )\big)
  + \dimc\big( \tft_\calc( \mathrm T; \epsilon )\big) \,\big]
  = 2\, (3 + 1) = 8 \,.
  \eear
  \ee
To see that these numbers agree with formula \erf{Vf00N}, just note that for
the Ising model we have $(S_{0,n}) \eq \frac12\,(1,\sqrt2,1)$ and 
$(S_{\sigma,n}) \eq \frac12\,(\sqrt2,0,{-}\sqrt2)$. 
The corresponding numbers for arbitrary $N_1$ are
  \be
  \mathrm d_{N_1}(\one_\calc) = 2^{N_1-3} + 2^{N_1/2-2} \qquad{\rm and}\qquad
  \mathrm d_{N_1}(\sigma) =  2^{3N_1/2-3} , 
  \ee
respectively, and thus grow exponentially with the number of genons. The growth
depends, however, explicitly on the choice of label for the genon, and a 
judicious choice leads to more powerful codes.

\paragraph{Braiding.}

For a $\zet_2$-equivariant category, the notion of a braiding has to be
replaced by the notion of an \emph{equivariant braiding}. Concretely, part 
of the data are two autoequivalences $\tau_\cald\colon \cald\To\cald$ and 
$\tau_\calc\colon \calc\To\calc$ of the categories involved. As shown in
\cite{barSc}, $\tau_\cald$ act as $\tau_\cald(U\boti V) \eq V\boti U$,
while $\tau_\calc$ is the identity endofunctor of \calc. If $V\iN\calc$ 
labels a circle with an odd number of \P-defects, then the equivariant braiding 
is given by functorial isomorphisms $c_{V,W}\colon V\oti W\To \tau(W)\oti V$. 
These isomorphisms have been computed \cite[Sect.\,4.6]{barSc} from the cover 
functor and look as follows.
  \def\leftmargini{1.27em}~\\[-1.25em]\begin{itemize}\addtolength{\itemsep}{-7pt}

\item
The braiding of two objects $V\eq V_1\boti V_2\iN\cald$ and 
$W \eq W_1\boti W_2\iN\cald$ is just $c^\calc_{V_1,W_1} \boti c^\calc_{V_2,W_2}$
\cite[Eq.\,(23)]{barSc}, as one would expect for bilayer systems.

\item
The equivariant braiding of $V \eq V_1\boti V_2\iN\cald$ and
$W\iN\calc$ is more complicated; it involves the twist $\theta$ as well as
over- and underbraidings:
  \be
  c_{V,W}^{} = (c^\calc_{V_1,W} \oti \theta^{-1}_{V_2}) \circ
  (\id_{V_1}^{} \oti (c^{\calc}_{W,V_2})^{-1}_{}): \quad
  V_1\oti V_2\oti W\stackrel\cong\to W\oti V_1\oti V_2 \,.
  \ee

\item
Similarly, for $V\iN \calc$ and $W \eq W_1\boti W_2\iN\cald$ we have
  \be
  c_{V,W}^{} = (c^\calc_{W_1,W_2}\oti \id_V^{}) \circ
  ((c^{\calc}_{W_1,V})^{-1}_{} \oti \id_{W_2}^{}): \quad 
  V\oti W_1\oti W_2 \stackrel\cong\to W_2\oti W_1\oti V \,.
  \ee

\item
The equivariant braiding of two objects in \cald\ is still more
complicated; we refer to the last equation-picture in \cite[Sect.\,4.6]{barSc}.
\end{itemize}

It is an interesting and important problem to obtain the appropriate
generalizations of mapping class groups of surfaces with 
boundary disks and defect lines and the representations of these groups on 
the spaces of generalized conformal blocks, as well as to relate them to
representations of mapping class groups of higher genus surfaces.
For results in the condensed matter literature see \cite{bajQ2,fhnqww}.
This issue is of direct relevance for the problem of implementing 
universal quantum gates on topological codes described by these spaces of 
generalized conformal blocks. It has already been demonstrated
\cite{frnw,bajQ2} that the presence of twist defects, via the
induced mapping class group actions of higher genus, renders the
double layer Ising system with permutation twist defects universal for 
quantum computing, while it is non-universal without defects in genus zero.

\paragraph{A remark on orbifolding.}

We conclude this note with a speculative remark. As pointed out in
\cite{bajQ2}, gauging the symmetry that underlies a twist defect can
deconfine the extrinsic defects such that they become intrinsic quasi-particles
in a topological phase described by the corresponding orbifold theory. As a 
physical mechanism for such a gauging, based on the analogy with the emergence 
of a $\zet_2$-gauge theory by a proliferation of double vortices in a 
superfluid the authors of \cite{bajQ2} propose a proliferation of 
double-twist defects. We conjecture that on the level of 
topological field theory this mechanism is implemented by a three-dimensional 
analogue of the generalized orbifolds of \cite{ffrs6}, in which an 
orbifold construction is realized with the help of a network of defect lines.


\newpage

{\small
\noindent{\sc Acknowledgments:}
JF is still to some extent supported 
by VR under project no.\ 621-2009-3993. CS is partially supported 
by the Collaborative Research Centre 676 ``Particles, Strings and the Early 
Universe - the Structure of Matter and Space-Time'' and by the DFG Priority 
Programme 1388 ``Representation Theory''. JF is grateful to Hamburg 
University, and in particular to CS, Astrid D\"orh\"ofer and Eva Kuhlmann, 
for their hospitality when part of this work was done.
}

 \vskip 4em 

 \newcommand\wb{\,\linebreak[0]} \def\wB {$\,$\wb}
 \newcommand\Bi[2]    {\bibitem[#2]{#1}} 
 \newcommand\inBO[9]  {{\em #9}, in:\ {\em #1}, {#2}\ ({#3}, {#4} {#5}), p.\ {#6--#7} {\tt [#8]}}
 \newcommand\J[7]     {{\em #7}, {#1} {#2} ({#3}) {#4--#5} {{\tt [#6]}}}
 \newcommand\JO[6]    {{\em #6}, {#1} {#2} ({#3}) {#4--#5} }
 \newcommand\JP[7]    {{\em #7}, {#1} ({#3}) {{\tt [#6]}}}
 \newcommand\BOOK[4]  {{\em #1\/} ({#2}, {#3} {#4})}
 \newcommand\PhD[2]   {{\em #2}, Ph.D.\ thesis #1}
 \newcommand\Prep[2]  {{\em #2}, preprint {\tt #1}}
 \def\aagt  {Alg.\wB\&\wB Geom.\wb Topol.}     
 \def\adma  {Adv.\wb Math.}
 \def\atmp  {Adv.\wb Theor.\wb Math.\wb Phys.}   
 \def\anma  {Ann.\wb Math.}
 \def\comp  {Com\-mun.\wb Math.\wb Phys.}
 \def\ijmp  {Int.\wb J.\wb Mod.\wb Phys.\ A}
 \def\imrn  {Int.\wb Math.\wb Res.\wb Notices}
 \def\jgap  {J.\wb Geom.\wB and\wB Phys.}
 \def\jomp  {J.\wb Math.\wb Phys.}
 \def\jhrs  {J.\wB Homotopy\wB Relat.\wB Struct.}
 \def\joal  {J.\wB Al\-ge\-bra}
 \def\jpaa  {J.\wB Pure\wB Appl.\wb Alg.}
 \def\jram  {J.\wB rei\-ne\wB an\-gew.\wb Math.}
 \def\leni  {Lenin\-grad\wB Math.\wb J.}
 \def\nupb  {Nucl.\wb Phys.\ B}
 \def\phlb  {Phys.\wb Lett.\ B}
 \def\quto  {Quantum Topology}
 \def\phrb  {Phys.\wb Rev.\ B}
 \def\phrx  {Phys.\wb Rev.\ X}
 \def\sema  {Selecta\wB Mathematica}
 \def\trgr  {Trans\-form.\wB Groups}

\end{document}